\documentclass[a4paper]{article}
\usepackage{graphicx}
\usepackage{amsmath,amsfonts}
\usepackage{amssymb}
\usepackage{color}
\usepackage{algorithm}
\usepackage{algpseudocode}
\usepackage{booktabs}
\usepackage{url}
\usepackage[subrefformat=parens]{subcaption}
\newenvironment{proof}[1][Proof]{\textbf{#1.} }{\hfill$\square$}
\newtheorem{theorem}{Theorem}
\newtheorem{prob}{Problem}

\newtheorem{remark}{Remark}
\newtheorem{ass}{Assumption}

\newcommand{\A}{{\mathcal A}}
\newcommand{\B}{{\mathcal B}}

\newcommand{\C}{{\mathcal C}}
\newcommand{\D}{{\mathcal D}}

\begin{document}

%\begin{frontmatter}

\title{Closed-Loop Identification of Periodically Time-Varying Systems via Cyclic Reformulation}%\thanksref{footnoteinfo}}

%\thanks[footnoteinfo]{This paper was not presented at any IFAC meeting. Corresponding author H.~Okajima.}

\author{Hiroshi Okajima}%\ead{okajima@cs.kumamoto-u.ac.jp}

%\address[Kumamoto]{Faculty of Advanced Science and Technology, Kumamoto University, Japan}

%\begin{keyword}
%Closed-loop identification; Cyclic reformulation; Joint input-output method; Periodically time-varying systems; Subspace identification.
%\end{keyword}

\begin{abstract}
This paper studies closed-loop identification of linear periodically
time-varying (LPTV) plants, with emphasis on open-loop unstable plants
for which open-loop experiments are not practically available. The
central contribution is an exact algebraic plant-extraction theorem for
cycled closed-loop realizations: for square strictly proper plants and a controller path satisfying an invertibility condition, the cycled plant transfer matrix is recovered from a shared state-space realization of the stable closed-loop maps from the external reference to the plant output and to the control input, without state augmentation, and without requiring the recovered plant realization to be stable. Thus, the stability requirement for data generation is shifted from the open-loop plant to the internally stable closed-loop system. Building on this result, a closed-loop identification algorithm is constructed that takes the reference, output, and input signals as data, applies standard subspace identification to the cycled signals, performs the algebraic plant extraction, and recovers the LPTV plant state-space parameters via a coordinate transformation; the conditioning of the inverse controller path governs the reliability of the extraction step. Numerical examples demonstrate the recovery of stable and open-loop unstable SISO LPTV plants and validate a MIMO case through coordinate-invariant Markov-parameter comparisons.
\end{abstract}

%\end{frontmatter}

\maketitle

%================================================================================
\section{Introduction}\label{sec:intro}
%================================================================================
System identification is a fundamental step in constructing mathematical models of dynamic systems for control system design \cite{id00,id01}. Subspace identification methods \cite{n4sid,moesp,katayama_book,overschee_book} have gained significant attention for state-space model identification of linear time-invariant (LTI) systems, owing to their numerical stability and natural applicability to multi-input multi-output (MIMO) systems.

LPTV systems arise naturally in various engineering domains, including multirate sampled-data systems \cite{cyc5,P2}, helicopter rotor dynamics, satellite and space tether attitude control, periodic disturbance rejection, and multi-agent systems with periodically varying networks \cite{P3} (see \cite{cyc2} for a comprehensive treatment). Identification methods for LPTV systems have been studied from several perspectives. Representative state-space approaches include ensemble methods using multiple datasets \cite{id0}, periodic-input methods for time-varying FIR or state-space models \cite{id1,id2,id11tac}, continuous-time identification of periodically parameter-varying models \cite{id3}, and lifting-based methods for autonomous periodic systems \cite{real}. The cyclic-reformulation-based approach of \cite{okajima_access2025} applies subspace identification to cycled open-loop signals to obtain an equivalent LTI state-space model, and then exploits structural properties of the Markov parameters to recover the original LPTV parameters via a state coordinate transformation. A related but distinct line of research treats linear parameter-varying (LPV) systems with periodic scheduling, including subspace identification of MIMO LPV systems~\cite{id11auto} and a unified LPV subspace identification framework that covers both open-loop and closed-loop settings~\cite{cox2021auto}. The LPV setting differs from the LPTV case in that parameter variations are driven by an external scheduling signal rather than by time periodicity, and the structural properties exploited in LPTV identification do not carry over directly.

This paper addresses the identification of open-loop unstable LPTV systems under closed-loop operation. The existing LPTV identification methods reviewed above are primarily formulated for open-loop input-output data, and their direct use is therefore not applicable to open-loop unstable plants. Such plants arise in practical applications such as spacecraft attitude control and certain rotating machinery, where the plant output diverges without feedback and closed-loop data are therefore indispensable. For LTI systems, closed-loop identification has a long history covering frameworks applicable to unstable plants, with early foundations~\cite{verhaegen1993} and subsequent surveys on the interplay between identification and control~\cite{vandenhof_schrama1995,forssell_ljung1999}. Among various closed-loop identification philosophies, the joint input-output approach of Katayama, Kawauchi, and Picci~\cite{katayama_kawauchi_picci2005} and related treatments and extensions~\cite{huang_kadali2008,oku_closedloop} identify the stable closed-loop map from the external reference to the joint signal $[y^{\top}\;u^{\top}]^{\top}$, and the consistency of closed-loop subspace methods under feedback has been analyzed in~\cite{chiuso_picci2005,chiuso2007auto}; closed-loop subspace identification of LTI systems including unstable plants remains an active area of research (see, e.g., \cite{id9}). However, a realization-theoretic route that starts from closed-loop data and explicitly recovers a finite-dimensional state-space LPTV plant model appears to remain underdeveloped. This gap is particularly relevant for open-loop unstable LPTV plants, for which open-loop experiments are not a viable data-generation mechanism.

To address this gap, this paper develops a closed-loop identification framework for LPTV systems. A closely related line of work on LPTV or multirate dynamics under feedback is the lifted frequency-domain identification of closed-loop multirate systems by van~Haren et al.~\cite{vanharen2025mechatronics}. That method targets nonparametric frequency-response estimation of lifted multirate systems, whereas the present paper recovers a finite-dimensional state-space LPTV plant model. By combining the joint input-output approach~\cite{katayama_kawauchi_picci2005,huang_kadali2008,oku_closedloop} with cyclic reformulation~\cite{cyc1,cyc2}, this paper separates the exact realization-theoretic extraction result from its finite-sample implementation by subspace identification, exploiting that the joint approach yields the two transfer matrices needed for plant extraction in a shared realization. The contributions are: (i) an exact algebraic theorem that recovers $\check{P}(z)=\check T_{yr}(z)\check T_{ur}(z)^{-1}$ from a shared-$A$ realization without state augmentation; (ii) validity of this extraction even when the recovered plant realization is unstable; (iii) explicit separation of the assumptions for data generation, identification, and algebraic inversion; and (iv) numerical validation for stable, open-loop unstable, and MIMO LPTV examples.

This paper is organized as follows. Section~\ref{sec:prelim} defines the LPTV plant and controller models and reviews the cyclic reformulation for a general LPTV system. Section~\ref{sec:closedloop} derives the closed-loop augmented system in the joint input-output framework. Section~\ref{sec:cyclic_cl} applies cyclic reformulation to the augmented closed-loop system. Section~\ref{sec:identification} describes the subspace identification procedure, the plant extraction method, and the recovery of the LPTV plant parameters. Section~\ref{sec:simulation} presents numerical simulation results. Section~\ref{sec:conclusion} concludes the paper.

\textit{Notation:}
$I_n$ denotes the $n\times n$ identity matrix and $O_{n,m}$ the $n\times m$ zero matrix. The accent $\check{\cdot}$ marks signals and system matrices associated with cyclic reformulation, defined in Section~\ref{sec:cyclic_review}. Calligraphic letters such as $\A,\B,\C,\C_y,\C_u$ denote quantities obtained from the identified closed-loop realization, in contrast to the true cycled matrices marked by $\check{\cdot}$. For an $M$-periodic matrix sequence $\{A_k\}$, the monodromy matrix is defined as $\Phi := A_{M-1} A_{M-2} \cdots A_{0}$, and a subscript is used to indicate the underlying sequence (e.g., $\Phi_p$ for the plant and $\Phi_{cl}$ for the closed-loop system). The Frobenius norm is denoted by $\|\cdot\|_F$.

%================================================================================
\section{Preliminaries}\label{sec:prelim}
%================================================================================

\subsection{Periodically time-varying plant}\label{sec:plant}

An $n_p$-th order discrete-time LPTV plant $P$ with period $M$ is described by 
\begin{eqnarray}
x_p(k+1)&=&A_{p,k}\, x_p(k)+B_{p,k}\, u(k) + B_{w,k}\, w(k)\label{eq:plant_state}\\
y(k) &=& C_{p,k}\, x_p(k) + D_{p,k}\, u(k) + v(k)\label{eq:plant_output}
\end{eqnarray}
where $u(k)\in \mathbb{R}^{m}$ is the control input, $x_p(k)\in \mathbb{R}^{n_p}$ is the plant state, $y(k)\in \mathbb{R}^{l}$ is the measured output, $w(k)\in \mathbb{R}^{m_w}$ is the process disturbance, and $v(k)\in \mathbb{R}^{l}$ is the measurement noise. The system matrices satisfy $A_{p,k} = A_{p,k+M}$, $B_{p,k} = B_{p,k+M}$, $B_{w,k} = B_{w,k+M}$, $C_{p,k} = C_{p,k+M}$, and $D_{p,k} = D_{p,k+M}$ for all $k$. By this periodicity, all system matrices are uniquely determined by their values at $k = 0, 1, \ldots, M-1$, and throughout this paper the subscript $k$ in a periodic matrix denotes the representative index $k\,\mathrm{mod}\,M$ unless otherwise stated. The pair $(C_{p,k}, A_{p,k})$ is assumed to be observable and the pair $(A_{p,k}, B_{p,k})$ is assumed to be controllable in the sense of periodic systems \cite{cyc2}.

\begin{ass}\label{ass:Dp_zero}
The plant is square ($m = l$) and has no direct feedthrough, i.e., $D_{p,k} = O_{l,m}$ for all $k$.
\end{ass}
No assumption is placed on the open-loop stability of the plant, and the eigenvalues of the monodromy matrix $\Phi_p := A_{p,M-1}\cdots A_{p,0}$ are allowed to lie outside the unit circle.

\subsection{Periodically time-varying controller}\label{sec:controller}

An $n_c$-th order discrete-time LPTV controller $K$ with the same period $M$ is described by
\begin{eqnarray}
x_c(k+1) &=& A_{c,k}\, x_c(k) + B_{c,k}\, e(k) \label{eq:ctrl_state}\\
u(k) &=& C_{c,k}\, x_c(k) + D_{c,k}\, e(k) \label{eq:ctrl_output}
\end{eqnarray}
where $x_c(k) \in \mathbb{R}^{n_c}$ is the controller state, $e(k) = r(k) - y(k) \in \mathbb{R}^{l}$ is the tracking error, and $r(k) \in \mathbb{R}^{l}$ is the reference signal. The system matrices satisfy $A_{c,k} = A_{c,k+M}$, $B_{c,k} = B_{c,k+M}$, $C_{c,k} = C_{c,k+M}$, and $D_{c,k} = D_{c,k+M}$ for all $k$.

\begin{ass}\label{ass:Dc_zero}
The controller has no direct feedthrough, i.e., $D_{c,k} = O_{m,l}$ for all $k$, and has relative degree one, i.e., $C_{c,k}\, B_{c,k-1}$ is nonsingular for all $k = 0, \ldots, M-1$.
\end{ass}

The zero direct feedthrough condition is standard in many practical control implementations. Note that under Assumption~\ref{ass:Dp_zero}, the controller input $e(k)$ and output $u(k)$ are of the same dimension $l = m$, so $C_{c,k}\, B_{c,k-1}$ is a square matrix and its nonsingularity is well-defined; for the SISO case ($m = l = 1$), this nonsingularity reduces to $C_{c,k}\, B_{c,k-1} \neq 0$ for all $k$. Since the controller is chosen by the designer, both conditions can be verified by construction. The relative degree one case is assumed for simplicity; the extension to relative degree $d \geq 2$, where $C_{c,k}\, A_{c,k-1}\cdots A_{c,k-d+1}\, B_{c,k-d}$ is nonsingular for all $k$, is given in Remark~\ref{rem:relative_degree_general}.

\subsection{Cyclic reformulation}\label{sec:cyclic_review}

The cyclic reformulation \cite{cyc1,cyc2} is a technique to convert an $M$-periodic LPTV system into an equivalent LTI system. The general framework is reviewed here.

Consider a general $n$-th order $M$-periodic LPTV system:
\begin{eqnarray}
x(k+1) &=& A_k\, x(k) + B_k\, u(k) \label{eq:general_lptv_state}\\
y(k) &=& C_k\, x(k) + D_k\, u(k) \label{eq:general_lptv_output}
\end{eqnarray}
where $x(k) \in \mathbb{R}^n$, $u(k) \in \mathbb{R}^{m_u}$, $y(k) \in \mathbb{R}^{m_y}$, and $A_k = A_{k+M}$, $B_k = B_{k+M}$, $C_k = C_{k+M}$, $D_k = D_{k+M}$.

The cycled input signal $\check{u}(k) \in \mathbb{R}^{Mm_u}$ is defined by
\begin{eqnarray}
\lefteqn{\check u(0) = \begin{bmatrix}u(0)\\O_{m_u,1}\\ \vdots\\ O_{m_u,1}\end{bmatrix}, \check u(1) = \begin{bmatrix}O_{m_u,1}\\u(1)\\ \vdots\\ O_{m_u,1}\end{bmatrix}, \cdots, }\label{eq:checku}\\&\check u(M\!-\!1) = \begin{bmatrix}O_{m_u,1}\\ \vdots\\O_{m_u,1}\\ u(M\!-\!1)\end{bmatrix}, \check u(M) = \begin{bmatrix}u(M)\\O_{m_u,1}\\ \vdots\\ O_{m_u,1}\end{bmatrix}, \cdots \nonumber
\end{eqnarray}
That is, $\check{u}(k)$ has a unique nonzero sub-vector $u(k)$ at the position determined by $k\,\mathrm{mod}\,M$, and all other sub-vectors are zero. The sub-vector cyclically shifts along the block components. The cycled output $\check{y}(k) \in \mathbb{R}^{Mm_y}$ is defined in the same manner.

The cyclic reformulation of (\ref{eq:general_lptv_state}), (\ref{eq:general_lptv_output}) yields the following LTI system:
\begin{eqnarray}
\check{x}(k+1) &=& \check{A}\, \check{x}(k) + \check{B}\, \check{u}(k) \label{eq:general_cycled_state}\\
\check{y}(k) &=& \check{C}\, \check{x}(k) + \check{D}\, \check{u}(k) \label{eq:general_cycled_output}
\end{eqnarray}
where $\check{x}(k) \in \mathbb{R}^{Mn}$ is the cycled state, and
\begin{equation}
     \check{A} = \begin{bmatrix}
          0 & 0 & \cdots & 0 & A_{M-1}\\
          A_{0} & 0 & \cdots & 0 & 0\\
          0 & A_{1} & \ddots & \vdots & \vdots\\
          \vdots & \ddots & \ddots & 0 & \vdots\\
          0 & \cdots & 0 & A_{M-2} & 0 
     \end{bmatrix}\label{eq:checkA_general}
\end{equation}
\begin{equation}
     \check{B} = \begin{bmatrix}
          0 & 0 & \cdots & 0 & B_{M-1}\\
          B_{0} & 0 & \cdots & 0 & 0\\
          0 & B_{1} & \ddots & \vdots & \vdots\\
          \vdots & \ddots & \ddots & 0 & \vdots\\
          0 & \cdots & 0 & B_{M-2} & 0 
     \end{bmatrix}\label{eq:checkB_general}
\end{equation}
\begin{equation}
     \check{C} = \mathrm{diag}\{C_{0},\, C_{1},\, \ldots,\, C_{M-1}\} \label{eq:checkC_general}
\end{equation}
\begin{equation}
     \check{D} = \mathrm{diag}\{D_{0},\, D_{1},\, \ldots,\, D_{M-1}\} \label{eq:checkD_general}
\end{equation}
The matrices $\check{A} \in \mathbb{R}^{Mn \times Mn}$ and $\check{B} \in \mathbb{R}^{Mn \times Mm_u}$ have the characteristic cyclic structure, while $\check{C} \in \mathbb{R}^{Mm_y \times Mn}$ and $\check{D} \in \mathbb{R}^{Mm_y \times Mm_u}$ are block-diagonal matrices. A key property of this construction is that, when the cycled input $\check{u}(k)$ is driven by (\ref{eq:checku}), the nonzero sub-vector of the cycled state $\check{x}(k)$ at time $k$ coincides with the original state $x(k)$ of the periodic system (\ref{eq:general_lptv_state}), (\ref{eq:general_lptv_output}), and the same holds for the cycled output $\check{y}(k)$ and the original output $y(k)$; thus, the LTI system (\ref{eq:general_cycled_state}), (\ref{eq:general_cycled_output}) exactly reproduces the input-output behavior of the original LPTV system within the cycled signal framework.

The Markov-parameter characterization of cycled LPTV systems is described as follows \cite{okajima_access2025}. Let $\check H(0)=\check D$ and $\check H(i)=\check C\check A^{i-1}\check B$ for $i\ge 1$ denote the Markov parameters of (\ref{eq:general_cycled_state}), (\ref{eq:general_cycled_output}), and let $\check S_q\in\mathbb{R}^{Mq\times Mq}$ be the block cyclic shift matrix whose $(i,i+1)$-th block is $I_q$ for $i=1,\ldots,M-1$, whose $(M,1)$-th block is $I_q$, and whose other blocks are zero. Define the blocks $H_{k}^{(h)}\in\mathbb{R}^{m_y\times m_u}$ for $k=0,\ldots,M-1$ and $h\ge 0$ (all subscripts mod $M$) by $H_{k}^{(0)}=D_k$ and $H_{k}^{(h)}=C_{k+h}\,A_{k+h-1}\cdots A_{k+1}\,B_{k}$ for $h\ge 1$. While \cite{okajima_access2025} states a two-sided shift identity, only the left-sided form is needed in what follows: for all $h\ge 0$,
\begin{eqnarray}
\check S_{m_y}^{\,h}\check H(h) = \mathrm{diag}\{H_{0}^{(h)},\, H_{1}^{(h)},\, \ldots,\, H_{M-1}^{(h)}\}.\label{eq:shifted_markov_sparsity}
\end{eqnarray}
This shifted-Markov-parameter sparsity is an input-output property and is therefore independent of the particular state-space realization $(\check A,\check B,\check C,\check D)$. Hence it is inherited by any exact realization, including realizations related to the cyclic form by a similarity transformation.

%================================================================================
\section{Closed-loop augmented system}\label{sec:closedloop}
%================================================================================

In this section, the closed-loop augmented system is constructed in the joint input-output framework. This augmented system is an $M$-periodic LPTV system.

\subsection{Construction of the augmented state-space model}\label{sec:cl_formulation}

Consider the feedback system where the LPTV plant $P$ in (\ref{eq:plant_state}), (\ref{eq:plant_output}) is controlled by the LPTV controller $K$ in (\ref{eq:ctrl_state}), (\ref{eq:ctrl_output}). The tracking error is $e(k) = r(k) - y(k)$, where $r(k) \in \mathbb{R}^{l}$ is the reference signal. 

Substituting $e(k) = r(k) - y(k) = r(k) - C_{p,k} x_p(k)$ (using Assumption~\ref{ass:Dp_zero} and setting $w(k) = 0$ and $v(k) = 0$ for the deterministic part) and $u(k) = C_{c,k} x_c(k)$ (using Assumption~\ref{ass:Dc_zero}) into the state equations, and defining the augmented state 
\begin{eqnarray}
\xi(k) = \begin{bmatrix} x_p(k) \\ x_c(k)\end{bmatrix} \in \mathbb{R}^{n_{cl}}, \quad n_{cl} = n_p + n_c \label{eq:augmented_state}
\end{eqnarray}
the following closed-loop augmented system is obtained:
\begin{eqnarray}
\xi(k+1) &=& A_{cl,k}\, \xi(k) + B_{cl,k}\, r(k) \label{eq:cl_state}\\
z(k) := \begin{bmatrix} y(k) \\ u(k) \end{bmatrix} &=& C_{cl,k}\, \xi(k) \label{eq:cl_output}
\end{eqnarray}
where the closed-loop system matrices are given by
\begin{eqnarray}
A_{cl,k} &=& \begin{bmatrix} A_{p,k} & B_{p,k} C_{c,k} \\ -B_{c,k} C_{p,k} & A_{c,k} \end{bmatrix} \label{eq:Acl}\\
B_{cl,k} &=& \begin{bmatrix} O_{n_p, l} \\ B_{c,k} \end{bmatrix} \label{eq:Bcl}\\
C_{cl,k} &=& \begin{bmatrix} C_{y,k} \\ C_{u,k} \end{bmatrix} = \begin{bmatrix} C_{p,k} & O_{l, n_c} \\ O_{m, n_p} & C_{c,k} \end{bmatrix} \label{eq:Ccl}
\end{eqnarray}
where $A_{cl,k} \in \mathbb{R}^{n_{cl} \times n_{cl}}$, $B_{cl,k} \in \mathbb{R}^{n_{cl} \times l}$, $C_{cl,k} \in \mathbb{R}^{(l+m) \times n_{cl}}$. Here, $z(k) \in \mathbb{R}^{l+m}$ is the joint output, $C_{y,k}$ extracts the plant output $y(k)$ from the augmented state, and $C_{u,k}$ extracts the control input $u(k)$.

The augmented system (\ref{eq:cl_state}), (\ref{eq:cl_output}) is in the standard form of an $M$-periodic LPTV system with state $\xi(k) \in \mathbb{R}^{n_{cl}}$, input $r(k) \in \mathbb{R}^{l}$, output $z(k) \in \mathbb{R}^{l+m}$, and $M$-periodic system matrices. The direct feedthrough is $D_{cl,k} = O_{(l+m),l}$ for all $k$, which follows from Assumptions~\ref{ass:Dp_zero} and~\ref{ass:Dc_zero}. This is precisely the form (\ref{eq:general_lptv_state}), (\ref{eq:general_lptv_output}) to which the cyclic reformulation can be applied.

\begin{remark}\label{rem:shared_A}
The transfer relations from $r$ to $y$ and from $r$ to $u$ share the same state dynamics $(A_{cl,k}, B_{cl,k})$ and differ only in the output matrices $C_{y,k}$ and $C_{u,k}$.
\end{remark}

\begin{ass}\label{ass:stable}
The closed-loop system is internally stable, i.e., all eigenvalues of the monodromy matrix $\Phi_{cl} := A_{cl,M-1} A_{cl,M-2} \cdots A_{cl,0}$ lie strictly inside the unit circle. Equivalently, $\check{A}_{cl}$ is Schur stable, since the nonzero eigenvalues of $\check{A}_{cl}$ are the $M$-th roots of the eigenvalues of $\Phi_{cl}$ \cite{cyc2}.
\end{ass}

This stability requirement, which replaces the open-loop plant stability assumed in the open-loop formulation of \cite{okajima_access2025}, ensures that the closed-loop signals $\{r(k),y(k),u(k)\}$ are bounded and that the cycled augmented system used in the subsequent identification is Schur stable.

\subsection{Inclusion of disturbances}\label{sec:cl_disturbances}

When the process disturbance $w(k)$ and measurement noise $v(k)$ are present, the closed-loop augmented state equation becomes
\begin{eqnarray}
\xi(k+1) &=& A_{cl,k}\, \xi(k) + B_{cl,k}\, r(k) + B_{\eta,k}\, \eta(k) \label{eq:cl_state_noise}
\end{eqnarray}
where $\eta(k) = \begin{bmatrix} w(k)^{\top} & v(k)^{\top} \end{bmatrix}^{\top} \in \mathbb{R}^{m_w + l}$ is the combined disturbance vector, and
\begin{eqnarray}
B_{\eta,k} &=& \begin{bmatrix} B_{w,k} & O_{n_p, l} \\ O_{n_c, m_w} & -B_{c,k} \end{bmatrix} \in \mathbb{R}^{n_{cl} \times (m_w + l)} \label{eq:Beta}
\end{eqnarray}
The output equation including noise is given by
\begin{eqnarray}
z(k) = \begin{bmatrix} y(k) \\ u(k) \end{bmatrix} &=& C_{cl,k}\, \xi(k) + \begin{bmatrix} I_l \\ O_{m,l} \end{bmatrix} v(k) \label{eq:cl_output_noise}
\end{eqnarray}
Here, $w(k)$ enters through the plant state, whereas $v(k)$ affects the measured output and propagates to the controller state through feedback. In the identification step, $r(k)$ is the known external input and $\eta(k)$ is treated as an unknown stochastic disturbance. The required uncorrelatedness between $r(k)$ and $\eta(k)$ holds when the reference is externally generated independently of the noise.

%================================================================================
\section{Cyclic reformulation of the augmented closed-loop system}\label{sec:cyclic_cl}
%================================================================================

In this section, the cyclic reformulation reviewed in Section~\ref{sec:cyclic_review} is applied to the augmented closed-loop system derived in Section~\ref{sec:closedloop}. The augmented system (\ref{eq:cl_state}), (\ref{eq:cl_output}) is an $M$-periodic LPTV system with state dimension $n_{cl}$, input dimension $l$, and output dimension $l+m$. Therefore, the general cyclic reformulation framework (\ref{eq:general_cycled_state})--(\ref{eq:checkD_general}) is directly applied by the substitutions $n \leftarrow n_{cl}$, $m_u \leftarrow l$, $m_y \leftarrow l+m$, $A_k \leftarrow A_{cl,k}$, $B_k \leftarrow B_{cl,k}$, $C_k \leftarrow C_{cl,k}$, and $D_k \leftarrow O$.

\subsection{Cycled closed-loop signals}\label{sec:cycled_cl_signals}

The cycled signals are constructed from the individual closed-loop signals $r(k)$, $y(k)$, and $u(k)$ separately. Specifically, the cycled reference signal $\check{r}(k) \in \mathbb{R}^{Ml}$, the cycled plant output $\check{y}(k) \in \mathbb{R}^{Ml}$, and the cycled control input $\check{u}(k) \in \mathbb{R}^{Mm}$ are each constructed in the same manner as (\ref{eq:checku}). The combined cycled output used for identification is defined by stacking $\check{y}(k)$ and $\check{u}(k)$ vertically:
\begin{eqnarray}
\check{z}(k) := \begin{bmatrix} \check{y}(k) \\ \check{u}(k) \end{bmatrix} \in \mathbb{R}^{M(l+m)} \label{eq:cycled_z}
\end{eqnarray}
Note that $\check{z}(k)$ is \emph{not} obtained by directly cycling the joint signal $z(k) = \begin{bmatrix} y(k)^{\top} & u(k)^{\top} \end{bmatrix}^{\top} \in \mathbb{R}^{l+m}$; rather, $\check{y}(k)$ and $\check{u}(k)$ are cycled individually and then stacked, so the nonzero block arrangement differs from the standard cyclic structure (\ref{eq:checku}). The $y$-channels occupy the first $Ml$ rows and the $u$-channels occupy the last $Mm$ rows of $\check{z}(k)$. This ordering is equivalent to the standard cyclic output of $z(k)$ up to a fixed row permutation, and is important for the subsequent decomposition of the identified output matrix.

\subsection{State-space representation of the cycled augmented system}\label{sec:cycled_cl_ss}

Applying the cyclic reformulation to the augmented system (\ref{eq:cl_state}), (\ref{eq:cl_output}), the following LTI system is obtained:
\begin{eqnarray}
\check{\xi}(k+1) &=& \check{A}_{cl}\, \check{\xi}(k) + \check{B}_{cl}\, \check{r}(k) \label{eq:cycled_cl_state}\\
\check{z}(k) &=& \check{C}_{cl}\, \check{\xi}(k) \label{eq:cycled_cl_output}
\end{eqnarray}
where $\check{\xi}(k) \in \mathbb{R}^{Mn_{cl}}$ is the cycled augmented state, $\check{z}(k)$ is the combined cycled output defined in (\ref{eq:cycled_z}), and the system matrices are given by
\begin{equation}
     \check{A}_{cl} = \begin{bmatrix}
          0 & 0 & \cdots & 0 & A_{cl,M-1}\\
          A_{cl,0} & 0 & \cdots & 0 & 0\\
          0 & A_{cl,1} & \ddots & \vdots & \vdots\\
          \vdots & \ddots & \ddots & 0 & \vdots\\
          0 & \cdots & 0 & A_{cl,M-2} & 0 
     \end{bmatrix}\label{eq:checkAcl}
\end{equation}
\begin{equation}
     \check{B}_{cl} = \begin{bmatrix}
          0 & 0 & \cdots & 0 & B_{cl,M-1}\\
          B_{cl,0} & 0 & \cdots & 0 & 0\\
          0 & B_{cl,1} & \ddots & \vdots & \vdots\\
          \vdots & \ddots & \ddots & 0 & \vdots\\
          0 & \cdots & 0 & B_{cl,M-2} & 0 
     \end{bmatrix}\label{eq:checkBcl}
\end{equation}
\begin{eqnarray}
\check{C}_y &=& \mathrm{diag}\{C_{y,0},\, C_{y,1},\, \ldots,\, C_{y,M-1}\} \label{eq:checkCy}\\
\check{C}_u &=& \mathrm{diag}\{C_{u,0},\, C_{u,1},\, \ldots,\, C_{u,M-1}\} \label{eq:checkCu}\\
\check{C}_{cl} &:=& \begin{bmatrix} \check{C}_y \\ \check{C}_u \end{bmatrix} \in \mathbb{R}^{M(l+m) \times Mn_{cl}} \label{eq:checkCcl}
\end{eqnarray}
where $\check{C}_y \in \mathbb{R}^{Ml \times Mn_{cl}}$ and $\check{C}_u \in \mathbb{R}^{Mm \times Mn_{cl}}$, and $\check{D}_{cl} = O$ since $D_{cl,k} = O$ for all $k$. Here, $C_{y,k} = \begin{bmatrix} C_{p,k} & O_{l,n_c} \end{bmatrix}$ and $C_{u,k} = \begin{bmatrix} O_{m,n_p} & C_{c,k} \end{bmatrix}$ as defined in (\ref{eq:Ccl}).

The dimensions of the remaining matrices are $\check{A}_{cl} \in \mathbb{R}^{Mn_{cl} \times Mn_{cl}}$ and $\check{B}_{cl} \in \mathbb{R}^{Mn_{cl} \times Ml}$. Note that $\check{A}_{cl}$ is built from the augmented closed-loop matrices $A_{cl,k}$ as a whole rather than from the plant and controller being cycled separately.

The cycled augmented system (\ref{eq:cycled_cl_state}), (\ref{eq:cycled_cl_output}) is an LTI system with $Ml$ inputs and $Ml + Mm$ outputs, and the state dimension is $Mn_{cl} = M(n_p + n_c)$. The transfer functions from $\check{r}$ to $\check{y}$ and from $\check{r}$ to $\check{u}$,
\begin{eqnarray}
\check{T}_{yr}(z) &=& \check{C}_{y}(zI - \check{A}_{cl})^{-1} \check{B}_{cl} \label{eq:Tyf_cycled}\\
\check{T}_{ur}(z) &=& \check{C}_{u}(zI - \check{A}_{cl})^{-1} \check{B}_{cl} \label{eq:Tuf_cycled}
\end{eqnarray}
share the same state dynamics $(\check{A}_{cl}, \check{B}_{cl})$ and differ only in the output matrices $\check{C}_y$ and $\check{C}_u$. This shared-$A$ structure is inherited from the original augmented system (Remark~\ref{rem:shared_A}) and preserved by cyclic reformulation.

\subsection{Inclusion of disturbances in the cycled system}\label{sec:cycled_disturbances}

When disturbances are present, the cyclic reformulation of (\ref{eq:cl_state_noise}), (\ref{eq:cl_output_noise}) yields
\begin{eqnarray}
\check{\xi}(k+1) &=& \check{A}_{cl}\, \check{\xi}(k) + \check{B}_{cl}\, \check{r}(k) + \check{B}_{\eta}\, \check{\eta}(k) \label{eq:cycled_cl_state_noise}\\
\begin{bmatrix} \check{y}(k) \\ \check{u}(k) \end{bmatrix} &=& \begin{bmatrix} \check{C}_y \\ \check{C}_u \end{bmatrix} \check{\xi}(k) + \begin{bmatrix} \check{D}_{vy} \\ O \end{bmatrix} \check{v}(k) \label{eq:cycled_cl_output_noise}
\end{eqnarray}
where $\check{\eta}(k)$ and $\check{v}(k)$ are the cycled disturbance signals, $\check{B}_{\eta}$ has the cyclic structure with sub-blocks $B_{\eta,k}$, and $\check{D}_{vy} = \mathrm{diag}\{I_l, \ldots, I_l\}$. In the identification procedure, $\check{r}(k)$ is the known input and the disturbance terms are treated as stochastic noise.

\begin{ass}\label{ass:obsv}
The cycled closed-loop realization $(\check A_{cl},\check B_{cl},\check C_{cl})$ is minimal, i.e., the pair $(\check A_{cl},\check B_{cl})$ is reachable and the pair $(\check C_{cl},\check A_{cl})$ is observable.
\end{ass}

Assumption~\ref{ass:obsv} concerns the cycled closed-loop interconnection $(\check A_{cl},\check B_{cl},\check C_{cl})$ as a whole, and is not implied by the plant-side controllability and observability assumed in Section~\ref{sec:plant}: the input enters only through $\check B_{cl}$, which involves $B_{c,k}$ but not $B_{p,k}$, and the output is the joint signal driven by the block-diagonal $\check C_{cl}$. The condition can be checked directly in computation by the rank of the reachability and observability matrices of $(\check A_{cl},\check B_{cl},\check C_{cl})$. Its role is to ensure that the nominal order $Mn_{cl}$ used in the subspace identification step coincides with the minimal order of the cycled closed-loop transfer matrix. When it fails, the same procedure still applies with $Mn_{cl}$ replaced by the effective minimal order, which is read off as a clear gap in the Hankel singular values.

%================================================================================
\section{Subspace identification and plant extraction}\label{sec:identification}
%================================================================================

This section formulates the closed-loop identification problem, applies subspace identification to the cycled signals, and then extracts the LPTV plant parameters from the identified closed-loop realization by algebraic formulas that are independent of any particular subspace identification algorithm.

\subsection{Identification problem}\label{sec:id_problem}

The corresponding open-loop problem --- recovering the LPTV plant parameters from open-loop input-output data $\{u(k), y(k)\}$ via cyclic reformulation and subspace identification --- has been solved in \cite{okajima_access2025}, where the open-loop data collection requires the plant to be open-loop stable. The closed-loop counterpart treated here generates the data under feedback, where the closed-loop system is internally stable by Assumption~\ref{ass:stable} but the plant itself may be open-loop unstable. This shifts the stability requirement for data generation from the open-loop plant to the closed-loop system.

\begin{prob}\label{prob:cl_id}
Given the closed-loop data $\{r(k), y(k), u(k)\}_{k=0}^{N-1}$, the period $M$, the plant order $n_p$, and the controller state-space realization $(A_{c,k}, B_{c,k}, C_{c,k}, D_{c,k})$, recover the LPTV plant parameters $A_{p,k}$, $B_{p,k}$, $C_{p,k}$, $D_{p,k}$ ($k = 0, \ldots, M-1$) of the plant (\ref{eq:plant_state}), (\ref{eq:plant_output}).
\end{prob}

The controller realization is used to specify the augmented order $Mn_{cl}=M(n_p+n_c)$ and to verify the structural assumptions, while the plant extraction formulas derived in Section~\ref{sec:plant_extraction} use only the identified shared closed-loop realization $(\A,\B,\C)$ with the partition $\C=[\C_y^{\top}\;\C_u^{\top}]^{\top}$.

The recovery target --- the LPTV plant parameters $\{A_{p,k}, B_{p,k}, C_{p,k}, D_{p,k}\}$ --- is the same as in the open-loop formulation of \cite{okajima_access2025}, but the input data are now the reference $r(k)$ rather than the manipulated input $u(k)$. The reference $r(k)$ is treated as a user-designed external signal that is sufficiently exciting for the cycled augmented system to be identified (e.g., a broadband signal such as a white noise sequence), and is uncorrelated with the process and measurement noises. As in the open-loop case, subspace identification determines the state-space realization only up to a similarity transformation, so the recovered LPTV parameters are determined up to a periodic coordinate transformation, which is an inherent degree of freedom in state-space identification.

\subsection{Application of subspace identification}\label{sec:subspace_id}

The closed-loop experiment proceeds in the original time domain: the external reference $r(k)$ is applied to the feedback loop, and the control input $u(k)$ and plant output $y(k)$ are recorded. From these recorded signals alone, the cycled signals $\check r(k)$ and $\check z(k)=[\check y(k)^{\top}\;\check u(k)^{\top}]^{\top}$ are constructed as defined in (\ref{eq:cycled_z}); this construction is a signal-level rearrangement and does not require any model knowledge. By the result of Section~\ref{sec:cyclic_cl}, the map from $\check r$ to $\check z$ generated by these signals is the LTI system (\ref{eq:cycled_cl_state}), (\ref{eq:cycled_cl_output}); the data $\{\check r(k),\check z(k)\}$ are therefore amenable to standard LTI subspace identification methods such as N4SID~\cite{n4sid} or MOESP~\cite{moesp}, applied with $\check r(k)\in\mathbb{R}^{Ml}$ as input and $\check z(k)\in\mathbb{R}^{M(l+m)}$ as output, under the constraint $\D=O$ implied by Assumptions~\ref{ass:Dp_zero} and~\ref{ass:Dc_zero}. The identified model is denoted by $(\A,\B,\C)$. This use of cyclic reformulation as a preprocessing step for subspace identification follows the open-loop formulation of~\cite{okajima_access2025}, while the joint stacking of $\check y$ and $\check u$ into $\check z$ is specific to the present closed-loop setting.

\begin{remark}\label{rem:similarity}
When the true cycled closed-loop realization is minimal and the identified realization is an exact minimal realization of the same deterministic transfer matrix, the two realizations are related by a similarity transformation $T_s \in \mathbb{R}^{Mn_{cl} \times Mn_{cl}}$:
\begin{eqnarray}
\A = T_s^{-1} \check{A}_{cl}\, T_s, \quad \B = T_s^{-1} \check{B}_{cl}, \quad \C = \begin{bmatrix} \check{C}_y \\ \check{C}_u \end{bmatrix} T_s . \label{eq:similarity}
\end{eqnarray}
This similarity relationship is used only when properties of the true structured realization, such as the nonsingularity of $\check C_u\check B_{cl}$, are transferred to the identified coordinates. The plant extraction formulas themselves do not require computing $T_s$.
\end{remark}

\subsection{Decomposition of the identified output matrix}\label{sec:output_decomp}

Since the cycled output $\check{z}(k) = \begin{bmatrix} \check{y}(k)^{\top} & \check{u}(k)^{\top} \end{bmatrix}^{\top}$ is constructed by stacking $\check{y}(k)$ (first $Ml$ rows) and $\check{u}(k)$ (last $Mm$ rows), the identified output matrix $\C$ can be partitioned accordingly:
\begin{eqnarray}
\C = \begin{bmatrix} \C_y \\ \C_u \end{bmatrix} \label{eq:Cstar_decomp}
\end{eqnarray}
where $\C_y \in \mathbb{R}^{Ml \times Mn_{cl}}$ corresponds to the first $Ml$ rows (the $\check{y}$-output) and $\C_u \in \mathbb{R}^{Mm \times Mn_{cl}}$ corresponds to the last $Mm$ rows (the $\check{u}$-output). The two sub-matrices $\C_y$ and $\C_u$ share the same state dynamics $(\A, \B)$, so the transfer functions from $\check{r}$ to $\check{y}$ and from $\check{r}$ to $\check{u}$ are given by
\begin{eqnarray}
\check{T}_{yr}(z) &=& \C_y(zI - \A)^{-1} \B \label{eq:Tyf}\\
\check{T}_{ur}(z) &=& \C_u(zI - \A)^{-1} \B \label{eq:Tuf}
\end{eqnarray}
This shared-$A$ structure is the fundamental property that enables the plant extraction without state dimension increase. Note that $\C_y$ and $\C_u$ are in general dense matrices (not block-diagonal), but this does not affect the extraction procedure, which relies solely on the input-output behavior.

\subsection{Plant extraction from the identified closed-loop model}\label{sec:plant_extraction}

Once the closed-loop model $(\A, \B, \C)$ with the output partition $\C=[\C_y^{\top}\;\C_u^{\top}]^{\top}$ in (\ref{eq:Cstar_decomp}) is identified by the procedure of Section~\ref{sec:subspace_id} with the joint stacked output $\check z(k)$, the cycled plant model can be extracted by exploiting the shared-$A$ structure. Here, $\check{P}(z)$ denotes the transfer function of the LTI system obtained by applying cyclic reformulation to the original LPTV plant. Thus, although the original plant is periodically time-varying and does not admit an LTI transfer function in the usual sense, the input-output relation between the cycled signals $\check{u}$ and $\check{y}$ is represented by the LTI transfer function $\check{P}(z)$. The key observation is that this $\check{P}(z)$ satisfies
\begin{eqnarray}
\check{P}(z) = \check{T}_{yr}(z)\, \check{T}_{ur}(z)^{-1} \label{eq:G_extraction}
\end{eqnarray}
This relation follows from the preservation of the input-output structure under cyclic reformulation. In the original time domain, the plant maps $u$ to $y$ through (\ref{eq:plant_state})--(\ref{eq:plant_output}), which implies that the cycled signals satisfy $\check{y}(k) = \check{P}(z)\,\check{u}(k)$ in the $z$-domain. Since the closed-loop transfer functions from $\check{r}$ to $\check{y}$ and from $\check{r}$ to $\check{u}$ satisfy $\check{T}_{yr}(z) = \check{P}(z)\,\check{T}_{ur}(z)$, one obtains (\ref{eq:G_extraction}) provided $\check{T}_{ur}(z)$ admits a right inverse as a rational matrix; since $\check{T}_{ur}(z)$ is strictly proper, this right invertibility is realized concretely through the proper transfer function $\check{T}_{ur}^{(1)}(z)$, whose nonsingular direct feedthrough $\C_u\B$ (cf.\ Remark~\ref{rem:CuB_structure}) ensures pointwise invertibility for almost all $z$.
Since $\check{T}_{yr}$ and $\check{T}_{ur}$ share the same state dynamics $(\A, \B)$ and $\D = O$, the direct computation of (\ref{eq:G_extraction}) can be performed without increasing the state dimension, as shown below.

Since $\D = O$, the transfer functions $\check{T}_{yr}(z)$ and $\check{T}_{ur}(z)$ are both strictly proper, i.e., their zeroth Markov parameters vanish. Define the one-step advanced transfer functions
\begin{eqnarray}
\check{T}_{yr}^{(1)}(z) &:=& z\,\check{T}_{yr}(z) \nonumber\\
&=& \C_y(zI - \A)^{-1}\A\B + \C_y\B \label{eq:shifted_Tyf}\\
\check{T}_{ur}^{(1)}(z) &:=& z\,\check{T}_{ur}(z) \nonumber\\
&=& \C_u(zI - \A)^{-1}\A\B + \C_u\B \label{eq:shifted_Tuf}
\end{eqnarray}
which are proper transfer functions with state-space realizations $(\A,\, \A\B,\, \C_y,\, \C_y\B)$ and $(\A,\, \A\B,\, \C_u,\, \C_u\B)$, respectively. These share the same state dynamics $(\A, \A\B)$. Under Assumption~\ref{ass:Dc_zero}, the direct feedthrough $\C_u\B$ of $\check{T}_{ur}^{(1)}(z)$ is nonsingular (cf.\ Remark~\ref{rem:CuB_structure}), so $\check{T}_{ur}^{(1)}(z)$ is invertible and the relation $\check{P}(z) = \check{T}_{yr}(z)\,\check{T}_{ur}(z)^{-1} = \check{T}_{yr}^{(1)}(z)\,[\check{T}_{ur}^{(1)}(z)]^{-1}$ holds since the common factor $z$ cancels.

\begin{remark}\label{rem:CuB_structure}
The relative degree one condition in Assumption~\ref{ass:Dc_zero} is connected to the cycled system as follows. The matrix $\check{C}_u \check{B}_{cl}$ is block-cyclic: for $k=0,\ldots,M-1$, its only nonzero block in the $k$-th block column is $C_{u,k}B_{cl,k-1}=C_{c,k}B_{c,k-1}$, located at the $(k+1)$-th block row. Hence, the nonsingularity of $C_{c,k}B_{c,k-1}$ for all $k=0,\ldots,M-1$ ensures that $\check{C}_u \check{B}_{cl}$ is nonsingular, and therefore so is $\C_u\B$ in the exact realization up to similarity. This is precisely the condition required for the inverse system construction used in the plant extraction.
\end{remark}

Under Assumption~\ref{ass:Dc_zero}, define $\Lambda := (\C_u\B)^{-1}$. Then, the cycled plant model $\check{P}(z)$ admits the following $Mn_{cl}$-th order realization:
\begin{eqnarray}
\A_p &=& \A - \A\B \Lambda\, \C_u \label{eq:Ag}\\
\B_p &=& \A\B \Lambda \label{eq:Bg}\\
\C_p &=& \C_y(I - \B \Lambda\, \C_u) \label{eq:Cg}\\
\D_p &=& \C_y\, \B \Lambda \label{eq:Dg}
\end{eqnarray}
where (\ref{eq:Cg}) is equivalently written as $\C_p = \C_y - \D_p\C_u$ by using (\ref{eq:Dg}).

\begin{theorem}\label{thm:plant_extraction}
Suppose that $(\A,\B,\C)$ with $\D=O$ realizes the joint cycled closed-loop map from $\check r$ to $\check z=[\check y^{\top}\;\check u^{\top}]^{\top}$, with $\C$ partitioned as $\C=[\C_y^{\top}\;\C_u^{\top}]^{\top}$ in (\ref{eq:Cstar_decomp}) so that $\check T_{yr}(z)$ and $\check T_{ur}(z)$ are given by (\ref{eq:Tyf})--(\ref{eq:Tuf}), and that $\C_u\B$ is nonsingular. Then $(\A_p,\B_p,\C_p,\D_p)$ defined by (\ref{eq:Ag})--(\ref{eq:Dg}) is a state-space realization of $\check P(z)$, i.e.,
\begin{eqnarray}
\check{P}(z) = \C_p(zI - \A_p)^{-1}\B_p + \D_p. \nonumber
\end{eqnarray}
\end{theorem}

\begin{proof}
The proof proceeds in four steps: (i) construct the standard realization of the inverse system $[\check{T}_{ur}^{(1)}(z)]^{-1}$, (ii) form the standard cascade realization of $\check P(z)=\check{T}_{yr}^{(1)}(z)\,[\check{T}_{ur}^{(1)}(z)]^{-1}$ of state dimension $2Mn_{cl}$, (iii) introduce a similarity transformation that block-diagonalizes the cascade and decouples its two halves, and (iv) identify an $Mn_{cl}$-dimensional unreachable subspace, leaving the realization (\ref{eq:Ag})--(\ref{eq:Dg}).

\emph{Step 1 (Inverse-system realization).}
Since $\check{T}_{ur}^{(1)}(z)$ has the realization $(\A,\A\B,\C_u,\C_u\B)$ with nonsingular direct feedthrough $\C_u\B$ (guaranteed by Assumption~\ref{ass:Dc_zero} and Remark~\ref{rem:CuB_structure}), the inverse system admits the standard realization \cite{katayama_book}
\begin{eqnarray}
[\check{T}_{ur}^{(1)}(z)]^{-1} \sim (\A-\A\B\Lambda\C_u,\, \A\B\Lambda,\, -\Lambda\C_u,\, \Lambda),\nonumber
\end{eqnarray}
whose state and input matrices coincide with $\A_p$ and $\B_p$ in (\ref{eq:Ag})--(\ref{eq:Bg}).

\emph{Step 2 (Cascade realization).}
Let $\check\nu$ be the external input to the cascade realization of $\check{T}_{yr}^{(1)}(z)\,[\check{T}_{ur}^{(1)}(z)]^{-1}$, which corresponds to the cycled plant input. The signal flow is $\check\nu\to[\check T_{ur}^{(1)}(z)]^{-1}\to\check T_{yr}^{(1)}(z)\to\check y$. Let $\xi_1\in\mathbb R^{Mn_{cl}}$ be the state of $[\check T_{ur}^{(1)}(z)]^{-1}$ and $\xi_2\in\mathbb R^{Mn_{cl}}$ the state of $\check T_{yr}^{(1)}(z)$. The standard series-connection rule yields the $2Mn_{cl}$-dimensional cascade realization $(\A_{\rm cas},\B_{\rm cas},\C_{\rm cas},\D_{\rm cas})$ with
\begin{eqnarray}
\A_{\rm cas}&=&\begin{bmatrix}\A_p & 0\\ -\A\B\Lambda\C_u & \A\end{bmatrix},\quad
\B_{\rm cas}=\begin{bmatrix}\B_p\\ \B_p\end{bmatrix},\nonumber\\[2pt]
\C_{\rm cas}&=&\begin{bmatrix}-\C_y\B\Lambda\C_u & \C_y\end{bmatrix},\quad
\D_{\rm cas}=\C_y\B\Lambda=\D_p.\nonumber
\end{eqnarray}
Here, the $(2,1)$-block $-\A\B\Lambda\C_u$ of $\A_{\rm cas}$ couples the output $-\Lambda\C_u\,\xi_1$ of the inverse system into the input channel of $\check T_{yr}^{(1)}(z)$ via the input matrix $\A\B$; the input matrix has identical blocks because the inverse-system input matrix is $\B_p=\A\B\Lambda$ and the direct feedthrough of the inverse system is $\Lambda$, which feeds $\check\nu$ to $\check T_{yr}^{(1)}(z)$ through $\A\B\Lambda=\B_p$.

\emph{Step 3 (Block-diagonalizing similarity).}
Define the similarity transformation $\xi=T\eta$ with
\begin{eqnarray}
T=\begin{bmatrix}I & 0\\ I & I\end{bmatrix},\quad
T^{-1}=\begin{bmatrix}I & 0\\ -I & I\end{bmatrix},\nonumber
\end{eqnarray}
so that $\eta_1=\xi_1$ and $\eta_2=\xi_2-\xi_1$. Direct computation, using $\A=\A_p+\A\B\Lambda\C_u$ from (\ref{eq:Ag}), gives
\begin{eqnarray}
T^{-1}\A_{\rm cas}T = \begin{bmatrix}\A_p & 0\\ 0 & \A\end{bmatrix},\;
T^{-1}\B_{\rm cas} = \begin{bmatrix}\B_p\\ 0\end{bmatrix},&&\nonumber\\[2pt]
\C_{\rm cas}T = \begin{bmatrix}\C_y(I - \B\Lambda\C_u) & \C_y\end{bmatrix} = \begin{bmatrix}\C_p & \C_y\end{bmatrix},&&\nonumber
\end{eqnarray}
where the last equality uses (\ref{eq:Cg}).

\emph{Step 4 (Mode cancellation).}
In the new coordinates, the $\eta_2$-subsystem has zero input matrix and is therefore unreachable from $\check\nu$. Hence $\eta_2(k)\equiv 0$ under zero initial conditions, and the input-output behavior of the cascade is realized by the $\eta_1$-subsystem alone:
\begin{eqnarray}
\eta_1(k+1) &=& \A_p\,\eta_1(k) + \B_p\,\check\nu(k),\nonumber\\
\check y(k) &=& \C_p\,\eta_1(k) + \D_p\,\check\nu(k).\nonumber
\end{eqnarray}
This is precisely the realization (\ref{eq:Ag})--(\ref{eq:Dg}), so $\check P(z)=\C_p(zI-\A_p)^{-1}\B_p+\D_p$. The realization has state dimension $Mn_{cl}$, despite the fact that the naive cascade has dimension $2Mn_{cl}$; the reduction by $Mn_{cl}$ is exactly the $Mn_{cl}$-dimensional unreachable subspace identified above.
\end{proof}

The realization (\ref{eq:Ag})--(\ref{eq:Dg}) has state dimension $Mn_{cl}$, the same as that of the identified shared closed-loop realization $(\A,\B,\C)$; hence the algebraic extraction does not require any state augmentation, despite the fact that the naive cascade realization constructed in the proof has dimension $2Mn_{cl}$.

Theorem~\ref{thm:plant_extraction} is stated as a purely algebraic result: it assumes only that $(\A,\B,\C)$ with $\D=O$ realizes the joint cycled closed-loop map and that $\C_u\B$ is nonsingular, and does not require this realization to be minimal. In the closed-loop LPTV setting of Sections~\ref{sec:closedloop}--\ref{sec:cyclic_cl}, the relation $\check T_{yr}=\check P\check T_{ur}$ used implicitly in the theorem follows from the plant input-output relation under cyclic reformulation, as discussed around (\ref{eq:G_extraction}). Assumptions~\ref{ass:Dp_zero} and~\ref{ass:Dc_zero} are structural conditions ensuring $\D=O$ and the nonsingularity of $\C_u\B$ (via Remark~\ref{rem:CuB_structure}) for an exact realization similar to the true cycled realization, and thereby supply the algebraic hypotheses of Theorem~\ref{thm:plant_extraction}. Assumptions~\ref{ass:stable} and~\ref{ass:obsv}, together with the excitation of $r(k)$ stated in Section~\ref{sec:id_problem}, are sufficient conditions under which the required exact shared realization is obtained as the ideal limit of the finite-data identification step, shifting the stability requirement from the open-loop plant to the closed-loop data-generating system.

\begin{remark}\label{rem:Ag_stability}
The proof of Theorem~\ref{thm:plant_extraction} is purely algebraic and does not require $\A_p$ to be stable. This is a crucial property: when the plant is open-loop unstable, $\A_p$ may have eigenvalues outside the unit circle, yet the transfer function identity $\check{P}(z) = \C_p(zI - \A_p)^{-1}\B_p + \D_p$ and all Markov parameters are preserved exactly. The mechanism is pole-zero cancellation in the shared-$A$ structure: $\check{T}_{yr}^{(1)}(z)$ has poles at the eigenvalues of $\A$, the inverse system $[\check{T}_{ur}^{(1)}(z)]^{-1}$ has transmission zeros at the same eigenvalues, and their product cancels all $\A$-modes exactly, leaving a realization whose poles are determined solely by $\A_p$. Hence, when the plant is open-loop unstable, $\A$ retains all eigenvalues inside the unit circle (reflecting the stable closed-loop dynamics) while $\A_p$ may have unstable eigenvalues, and the proposed method is applicable to open-loop unstable plants without any modification.
\end{remark}

\begin{remark}\label{rem:Dg_zero}
The direct term $\D_p=\C_y\B\Lambda$ in (\ref{eq:Dg}) is retained as a general expression valid for any shared-$A$ realization with strictly proper $\check T_{yr}(z)$ and $\check T_{ur}(z)$; under Assumptions~\ref{ass:Dp_zero} and~\ref{ass:Dc_zero} it is structurally zero in the exact case, as shown below, but is kept in the formula so that the proof of Theorem~\ref{thm:plant_extraction} remains self-contained and the extraction step is well-defined under finite-data perturbation. Assumption~\ref{ass:Dp_zero} ($D_{p,k}=O$) is essential for the particular closed-loop output decomposition used in (\ref{eq:Ccl}). If $D_{p,k}\neq O$, then
\begin{eqnarray}
y(k)=C_{p,k}x_p(k)+D_{p,k}C_{c,k}x_c(k)
\end{eqnarray}
under Assumption~\ref{ass:Dc_zero}, so the $y$-output block of the augmented closed-loop realization would be $[C_{p,k}\; D_{p,k}C_{c,k}]$ rather than $[C_{p,k}\; O]$. The extraction formulas in (\ref{eq:Ag})--(\ref{eq:Dg}) are therefore stated for strictly proper plants. Under Assumptions~\ref{ass:Dp_zero} and~\ref{ass:Dc_zero}, $\check C_y\check B_{cl}=O$ structurally; hence $\C_y\B=O$ in an exact similar realization and $\D_p=\C_y\B\Lambda=O$, so the extracted plant $\check P(z)$ is strictly proper. In finite-data noisy identification, however, $\C_y\B$ and $\C_u\B$ are perturbed. Thus, the direct term $\D_p=\C_y\B\Lambda$ and the matrices in (\ref{eq:Ag})--(\ref{eq:Dg}) are sensitive to the conditioning of $\C_u\B$. In particular, the extraction step should be regarded as reliable only when the estimated matrix $\C_u\B$ is nonsingular and reasonably well-conditioned. By standard matrix perturbation theory, the error in $\Lambda=(\C_u\B)^{-1}$ is bounded in terms of $\|(\C_u\B)^{-1}\|$ multiplied by the perturbation $\|\Delta(\C_u\B)\|$, so the conditioning of $\C_u\B$ governs the local sensitivity of the extraction step.
\end{remark}

\begin{remark}\label{rem:relative_degree_general}
The construction extends to controller relative degree $d \geq 2$ by using the shifted transfer functions $\check{T}_{yr}^{(d)}(z) := z^d\,\check{T}_{yr}(z)$ and $\check{T}_{ur}^{(d)}(z) := z^d\,\check{T}_{ur}(z)$, and replacing $\A\B$, $\C_y\B$, and $\C_u\B$ in (\ref{eq:Ag})--(\ref{eq:Dg}) with $\A^d\B$, $\C_y\A^{d-1}\B$, and $\C_u\A^{d-1}\B$, respectively. The nonsingularity of $\C_u\A^{d-1}\B$ follows from the relative degree $d$ condition in Assumption~\ref{ass:Dc_zero} via the cyclic structure of $\check{C}_u \check{A}_{cl}^{d-1}\check{B}_{cl}$, and the proof of Theorem~\ref{thm:plant_extraction} carries over by the same cancellation mechanism.
\end{remark}

\subsection{Reduction and recovery of LPTV plant parameters}\label{sec:order_reduction}

The realization $(\A_p, \B_p, \C_p, \D_p)$ obtained in Theorem~\ref{thm:plant_extraction} has dimension $Mn_{cl}=M(n_p+n_c)$, while the cycled plant transfer matrix $\check{P}(z)$ has minimal order $Mn_p$ when the original LPTV plant is minimal in the periodic sense. The extra $Mn_c$ modes correspond to the controller dynamics that cancel in the product $\check{T}_{yr}(z)\check{T}_{ur}(z)^{-1}$. In the exact (infinite-data) case these modes appear with zero Hankel singular values; with finite noisy data they appear as small but nonzero singular values, and a clear gap between the $Mn_p$-th and $(Mn_p+1)$-th singular values serves as a numerical indicator that the reduction to order $Mn_p$ is well-posed. Balanced truncation \cite{overschee_book} reduces the realization to order $Mn_p$; the procedure remains applicable when $\A_p$ has unstable eigenvalues (Remark~\ref{rem:Ag_stability}), since the cancellable modes that need to be removed correspond to the controller dynamics and lie in the stable part of the realization. The reduced realization is again denoted by $(\A_p, \B_p, \C_p, \D_p)$, with $\C_p\in\mathbb{R}^{Ml\times Mn_p}$ for the square case $m=l$.

By Theorem~\ref{thm:plant_extraction}, the reduced realization realizes $\check P(z)$, so its Markov parameters $\D_p$ and $\C_p\A_p^{i-1}\B_p$ ($i\geq 1$) coincide with those of the true cycled plant in the exact case and inherit the shifted block-diagonal sparsity recalled in Section~\ref{sec:cyclic_review}; under finite-data noise the sparsity holds only approximately. This inheritance can also be verified algebraically by induction on the Markov-parameter index using $\check P(z)\check T_{ur}(z)=\check T_{yr}(z)$, the sparsity of $\C_y\A^{i-1}\B$ and $\C_u\A^{i-1}\B$, and the cyclic structure of $\Lambda$ from Remark~\ref{rem:CuB_structure}. The coordinate transformation
\begin{eqnarray}
T^{-1} = \sum_{j=1}^{n_p} \check{F}_j \check{S}_l^{\,j-1} \C_p\, \A_p^{\,j-1} \label{eq:T_recovery}
\end{eqnarray}
then converts the realization into cyclic reformulation form, where $\check S_l$ is the block cyclic shift matrix defined in Section~\ref{sec:cyclic_review} and $\check F_j=\mathrm{diag}\{F_j,\ldots,F_j\}\in\mathbb{R}^{Mn_p\times Ml}$ is built from a common block $F_j\in\mathbb{R}^{n_p\times l}$ following \cite{okajima_access2025}, with $F_j$ chosen so that $T$ is nonsingular. The construction of $T^{-1}$ in (\ref{eq:T_recovery}) and the resulting extraction of the LPTV parameters from the cyclic-form realization are adopted from the open-loop parameter recovery procedure of \cite{okajima_access2025}, which exploits the shifted block-diagonal sparsity of the cycled Markov parameters recalled in Section~\ref{sec:cyclic_review}. For SISO plants, taking $F_j=e_j$ yields the LPTV plant in observable canonical form. From the resulting cyclic-form realization $(\check{A}_p,\check{B}_p,\check{C}_p,\check{D}_p)$, the LPTV parameters $A_{p,k}$, $B_{p,k}$, $C_{p,k}$, $D_{p,k}$ are read off from the block components. The recovery step is thus algebraic in the exact case and a structure-guided numerical reconstruction under finite-data noise, whose consistency is checked through Markov-parameter errors and the conditioning of $T$.

\subsection{Identification algorithm}\label{sec:algorithm}

The complete identification procedure is summarized as follows.

\begin{algorithm}[!bth]
\caption{Closed-Loop Identification for LPTV Systems}
\label{algo:cl_id}
\begin{algorithmic}
\State [1.] Collect the closed-loop data $\{r(k), y(k), u(k)\}_{k=0}^{N-1}$.
\State [2.] Construct the cycled signals: cycled input $\check{r}(k) \in \mathbb{R}^{Ml}$ from $r(k)$, cycled plant output $\check{y}(k) \in \mathbb{R}^{Ml}$ from $y(k)$, and cycled control input $\check{u}(k) \in \mathbb{R}^{Mm}$ from $u(k)$. Form the combined cycled output $\check{z}(k) = \begin{bmatrix} \check{y}(k)^{\top} & \check{u}(k)^{\top} \end{bmatrix}^{\top} \in \mathbb{R}^{M(l+m)}$.
\State [3.] Apply a subspace identification method (e.g., N4SID) to the cycled input $\check{r}(k)$ and cycled output $\check{z}(k)$ with the model order $Mn_{cl} = M(n_p + n_c)$ and $\D = O$.
\State [4.] Obtain the identified state-space model $(\A, \B, \C)$ and decompose $\C$ into $\C_y$ and $\C_u$.
\State [5.] Verify that $\C_u\B$ is nonsingular and compute $\Lambda = (\C_u\B)^{-1}$.
\State [6.] Compute the cycled plant realization $(\A_p, \B_p, \C_p, \D_p)$ using (\ref{eq:Ag})--(\ref{eq:Dg}).
\State [7.] Reduce $(\A_p, \B_p, \C_p, \D_p)$ from order $Mn_{cl}$ to the cycled plant order $Mn_p$ by balanced truncation; the cancellable modes targeted by the reduction are stable and the procedure applies regardless of whether $\A_p$ has eigenvalues outside the unit circle (Remark~\ref{rem:Ag_stability}). Continue to denote the reduced realization by $(\A_p, \B_p, \C_p, \D_p)$.
\State [8.] Select block-diagonal matrices $\check F_j$ in (\ref{eq:T_recovery}), construct $T^{-1}$, verify that $T$ is nonsingular, and apply the coordinate transformation $T$ of \cite{okajima_access2025} to convert the reduced realization $(\A_p, \B_p, \C_p, \D_p)$ into cyclic reformulation form $(\check A_p, \check B_p, \check C_p, \check D_p)$.
\State [9.] Extract the LPTV parameters $A_{p,k}, B_{p,k}, C_{p,k}, D_{p,k}$ ($k=0,\ldots,M-1$) from the block components of the cyclic matrices $(\check A_p, \check B_p, \check C_p, \check D_p)$ according to (\ref{eq:checkA_general})--(\ref{eq:checkD_general}).
\end{algorithmic}
\end{algorithm}

In Algorithm~\ref{algo:cl_id} and Section~\ref{sec:order_reduction}, the matrices $(\A,\B,\C)$ with the partition $\C=[\C_y^{\top}\;\C_u^{\top}]^{\top}$, $(\A_p,\B_p,\C_p,\D_p)$ (before and after order reduction), the coordinate transformation $T$, the cyclic-form realization $(\check A_p,\check B_p,\check C_p,\check D_p)$, and the recovered LPTV parameters $A_{p,k},B_{p,k},C_{p,k},D_{p,k}$ all denote the quantities computed from finite noisy data; in the numerical examples of Section~\ref{sec:simulation} the recovered LPTV parameters are written as $\hat A_{p,k},\hat B_{p,k},\hat C_{p,k},\hat D_{p,k}$ to make the comparison with the true LPTV parameters explicit. The corresponding exact (infinite-data) quantities defined in Sections~\ref{sec:plant_extraction}--\ref{sec:order_reduction} are the limits of these data-based quantities under Assumptions~\ref{ass:stable} and~\ref{ass:obsv} together with the excitation of $r(k)$ stated in Section~\ref{sec:id_problem}, and the algebraic extraction step is locally continuous in a neighborhood of the exact realization, as discussed in Remark~\ref{rem:Dg_zero}. The only stochastic estimation step is the subspace identification of the cycled closed-loop realization, for which the employed method is assumed to provide an estimate converging to an exact minimal realization under its standard assumptions; conditional on this convergence, the subsequent extraction and recovery steps are algebraic. The decoration $\check{}$ on $(\check A_p,\check B_p,\check C_p,\check D_p)$ refers to the cyclic block structure of the realization, in line with the notation used for the true cycled closed-loop matrices such as $\check A_{cl}$.

%================================================================================
\section{Numerical simulation}\label{sec:simulation}
%================================================================================

The proposed closed-loop identification approach is demonstrated through three numerical examples: a stable SISO plant (Example~1), an open-loop unstable SISO plant (Example~2), and a MIMO plant (Example~3). The source code reproducing the numerical results is available at \cite{code_repo}. In all examples, the N4SID method with MOESP weighting~\cite{n4sid} is used with $\D = O$, and the reference signal $r(k)$ is generated as an external white-noise sequence independent of the process and measurement noises. Noise-free results reproduce Theorem~\ref{thm:plant_extraction} at machine precision, while SNR $=40$ dB cases illustrate the same algebraic step applied to a realization estimated from finite noisy data. Example~1 reports numerical values in the text; Examples~2 and~3 summarize the corresponding results in tables. Since the recovered LPTV parameters are determined only up to a periodic similarity transformation, the Markov-parameter error is the primary, coordinate-invariant validation criterion throughout. In the SISO examples, the true plant is given in observable canonical form for $(A_{p,k},C_{p,k})$ and $\check F_j$ in (\ref{eq:T_recovery}) is chosen so that the recovery step reproduces this canonical form, fixing the periodic similarity ambiguity and making the entrywise errors of $(\hat A_{p,k},\hat B_{p,k},\hat C_{p,k})$ a well-defined supplementary diagnostic; the MIMO case (Section~\ref{sec:sim_mimo}) uses only coordinate-invariant criteria.

Two evaluation metrics are used throughout. The model fit, for a signal $\zeta(k)$ and its estimate $\hat{\zeta}(k)$ over a validation horizon, is $\mathrm{fit}[\%] = 100\left(1 - \|\zeta - \hat{\zeta}\|_2 / \|\zeta - \bar{\zeta}\|_2\right)$, where $\bar{\zeta}$ is the mean of $\zeta$. The closed-loop (CL) fit and open-loop (OL) fit denote this measure applied to the joint output $z = [y^{\top}\; u^{\top}]^{\top}$ and the plant output $y$, respectively, averaged over the output channels. The maximum Markov parameter error is $\max_{0 \leq h \leq h_{\max}} \|H(h) - \hat{H}(h)\|_F$ with $h_{\max} = 15$, where $H(h)$ and $\hat{H}(h)$ denote the $h$-th Markov parameter of the cycled plant and its estimate, respectively.

\subsection{Example 1: Stable SISO LPTV plant}\label{sec:sim_setup}

Consider the $2$nd-order SISO LPTV plant ($n_p = 2$, $m = l = 1$, $M = 3$) from \cite{okajima_access2025}:
\begin{eqnarray}
&A_{p,0} = \begin{bmatrix}0&1\\0.5&1\end{bmatrix},\; A_{p,1} = \begin{bmatrix}0&1\\0.9&-0.95\end{bmatrix},\; A_{p,2} = \begin{bmatrix}0&1\\1&0.5\end{bmatrix}, \nonumber \\
&B_{p,0} = \begin{bmatrix}1\\2\end{bmatrix},\; B_{p,1} = \begin{bmatrix}1.5\\2\end{bmatrix},\; B_{p,2} = \begin{bmatrix}1\\0.5\end{bmatrix}, \nonumber \\
&C_{p,k} = \begin{bmatrix}1&0\end{bmatrix},\quad D_{p,k} = 0, \quad k = 0, 1, 2 \nonumber
\end{eqnarray}
A time-invariant controller ($n_c = 1$) with $A_{c,k} = 0.30$, $B_{c,k} = 0.80$, $C_{c,k} = 0.05$, $D_{c,k} = 0$ stabilizes the closed-loop system ($\max|\mathrm{eig}(\Phi_{cl})| = 0.819$). The relative degree one condition in Assumption~\ref{ass:Dc_zero} holds since $C_{c,k}\, B_{c,k-1} = 0.04 \neq 0$ for all $k$. The cycled augmented system has state dimension $Mn_{cl} = 9$, input dimension $Ml = 3$, and output dimension $M(l+m) = 6$.

Using $N = 3000$ data points with SNR $= 40$ dB, the N4SID method with MOESP weighting achieved a closed-loop fit of $99.4\%$. The Hankel singular values of the extracted plant exhibit a clear gap at $Mn_p = 6$ ($\sigma_6 = 2.32$, $\sigma_7 = 3.17 \times 10^{-15}$), confirming that the $Mn_c = 3$ cancellable modes are numerically separated from the $Mn_p = 6$ plant modes. The cancellable modes are removed by balanced truncation, and the LPTV parameters are recovered via \cite{okajima_access2025} with $\check F_j$ chosen consistently with the observable canonical form of $(A_{p,k},C_{p,k})$, yielding an LPTV open-loop fit of $99.998\%$. Under this canonical alignment, the Frobenius-norm errors of the recovered LPTV matrices, averaged over $k=0,1,2$, are $\|A_{p,k}-\hat A_{p,k}\|_F=O(10^{-6})$, $\|B_{p,k}-\hat B_{p,k}\|_F=O(10^{-5})$, and $\|C_{p,k}-\hat C_{p,k}\|_F=O(10^{-16})$. The direct term is fixed at $\hat D_{p,k}=0$ by the strict-properness structure (Assumption~\ref{ass:Dp_zero}); the extraction quantity $\C_y\B\Lambda$ has Frobenius norm $O(10^{-2})$ under SNR $=40$ dB and decreases to machine precision in the noise-free case, consistent with Remark~\ref{rem:Dg_zero}. The noise-free part verifies Theorem~\ref{thm:plant_extraction} at machine precision, and the SNR $=40$ dB results show that Algorithm~\ref{algo:cl_id} solves Problem~\ref{prob:cl_id} for this stable SISO LPTV plant with the structural assumptions of the present framework.

\subsection{Example 2: Open-loop unstable SISO LPTV plant}\label{sec:sim_unstable}

To demonstrate the key advantage of the proposed method, an open-loop unstable LPTV plant is considered. The plant ($n_p = 2$, $m = l = 1$, $M = 3$) has matrices:
\begin{eqnarray}
&A_{p,0} = \begin{bmatrix}0&1\\0.8&1.2\end{bmatrix},\; A_{p,1} = \begin{bmatrix}0&1\\1.1&-0.5\end{bmatrix},\; A_{p,2} = \begin{bmatrix}0&1\\0.9&0.8\end{bmatrix}, \nonumber \\
&B_{p,0} = \begin{bmatrix}1\\2\end{bmatrix},\; B_{p,1} = \begin{bmatrix}1.5\\1\end{bmatrix},\; B_{p,2} = \begin{bmatrix}1\\1.5\end{bmatrix}, \nonumber \\
&C_{p,k} = \begin{bmatrix}1&0\end{bmatrix},\quad D_{p,k} = 0, \quad k = 0, 1, 2 \nonumber
\end{eqnarray}
The monodromy matrix eigenvalues are $-0.501$ and $1.581$, whose magnitudes exceed one for the second eigenvalue, confirming open-loop instability. The open-loop identification method of \cite{okajima_access2025} is therefore inapplicable.

A periodically time-varying controller ($n_c = 1$) with $A_{c,0} = 0.20$, $A_{c,1} = -0.50$, $A_{c,2} = 0.50$, $B_{c,k} = 0.80$, $C_{c,k} = 0.30$, $D_{c,k} = 0$ stabilizes the closed-loop system ($\max|\mathrm{eig}(\Phi_{cl})| = 0.564$). The relative degree one condition in Assumption~\ref{ass:Dc_zero} is verified as $C_{c,k}\, B_{c,k-1} = 0.24 \neq 0$ for all $k$.

Table~\ref{tab:unstable_results} summarizes the identification results using $N = 5000$ data points. In the noise-free case, the maximum Markov parameter error is at machine precision, and the extracted plant realization has $\max|\mathrm{eig}(\A_p)|=1.165$, matching $1.581^{1/3}\approx 1.1647$, the $M$-th root of the unstable monodromy eigenvalue of the true plant, consistent with Theorem~\ref{thm:plant_extraction} and Remark~\ref{rem:Ag_stability}. At SNR $= 40$ dB, the CL fit remains above $99\%$ and the Markov error degrades to $O(10^{-5})$, showing that the algebraic extraction is robust to moderate noise. The condition number of $\C_u\B$ stays at unity in both cases, indicating that the extraction step is numerically well posed.

\begin{table}[h]
\centering
\caption{Identification results for the unstable plant (Example~2)}
\label{tab:unstable_results}
\begin{tabular}{ccc}
\hline
 & SNR $= \infty$ & SNR $= 40$ dB \\
\hline
CL fit [\%] & $100.0$ & $99.3$ \\
Max Markov error & $3.78 \times 10^{-14}$ & $5.35 \times 10^{-5}$ \\
$\mathrm{cond}(\C_u\B)$ & $1.00$ & $1.00$ \\
\hline
\end{tabular}
\end{table}

The extracted plant realization has $Mn_{cl}=9$ states, of which three eigenvalues lie outside the unit circle (corresponding to the unstable plant modes) and the remaining six are inside the unit circle. The Hankel singular values exhibit a clear gap with six values at $O(10^0)$ and three at $O(10^{-16})$, the latter corresponding to the $Mn_c=3$ cancellable controller modes; combined with the three unstable modes, this recovers the expected plant order $Mn_p = 6$. Balanced truncation removes the cancellable modes, and the LPTV parameters are recovered with the same observable-canonical-form-aligned choice of $\check F_j$ as in Example~1. The entrywise recovery errors at SNR $=40$ dB are reported in Table~\ref{tab:unstable_lptv}. The results confirm that an open-loop unstable cycled plant realization is recovered from stable closed-loop data: the open-loop method of \cite{okajima_access2025} is inapplicable to this plant, while the proposed framework solves Problem~\ref{prob:cl_id} with the algebraic extraction of Theorem~\ref{thm:plant_extraction} remaining valid in the presence of unstable plant modes, as predicted by Remark~\ref{rem:Ag_stability}.

\begin{table}[h]
\centering
\caption{LPTV parameter recovery errors for the unstable plant (Example~2, SNR $= 40$ dB)}
\label{tab:unstable_lptv}
\begin{tabular}{cccc}
\hline
$k$ & $\|A_{p,k} - \hat{A}_{p,k}\|_F$ & $\|B_{p,k} - \hat{B}_{p,k}\|_F$ & $\|C_{p,k} - \hat{C}_{p,k}\|_F$ \\
\hline
$0$ & $1.64 \times 10^{-6}$ & $6.15 \times 10^{-6}$ & $0.00$ \\
$1$ & $8.53 \times 10^{-7}$ & $4.78 \times 10^{-6}$ & $2.78 \times 10^{-17}$ \\
$2$ & $1.31 \times 10^{-6}$ & $2.27 \times 10^{-6}$ & $2.78 \times 10^{-17}$ \\
\hline
\end{tabular}
\end{table}

\subsection{Example 3: MIMO LPTV plant}\label{sec:sim_mimo}

A $3$rd-order, $2$-input $2$-output LPTV plant ($n_p = 3$, $m = l = 2$, $M = 3$) is considered. The plant is constructed from a stable LTI base by periodically varying the $(3,3)$ entry of $A_p$:
\begin{eqnarray}
&A_{p,k} = \begin{bmatrix} 0 & 1 & 0 \\ 0 & 0 & 1 \\ -0.3 & -0.5 & a_{33,k} \end{bmatrix}, \nonumber \\
&(a_{33,0},\, a_{33,1},\, a_{33,2}) = (0.2,\, -0.1,\, 0.4)  \nonumber
\end{eqnarray}
\begin{eqnarray}
&B_{p,0} = \begin{bmatrix} 1.0 & 0.0 \\ 0.5 & 1.0 \\ 0.0 & 0.5 \end{bmatrix},\;
B_{p,1} = \begin{bmatrix} 1.0 & 0.2 \\ 0.3 & 1.0 \\ 0.1 & 0.4 \end{bmatrix},\;
B_{p,2} = \begin{bmatrix} 0.8 & 0.1 \\ 0.6 & 0.8 \\ 0.0 & 0.6 \end{bmatrix} \nonumber
\end{eqnarray}
with $C_{p,k} = \begin{bmatrix} 1 & 0 & 0 \\ 0 & 1 & 0 \end{bmatrix}$ and $D_{p,k} = O_{2,2}$. The plant monodromy matrix has $\max|\mathrm{eig}(\Phi_p)| = 0.629$. A $2$nd-order MIMO LPTV controller ($n_c = 2$) stabilizes the closed-loop system ($\max|\mathrm{eig}(\Phi_{cl})| = 0.862$). The relative degree one condition in Assumption~\ref{ass:Dc_zero} is verified: $\det(C_{c,k}\, B_{c,k-1}) \neq 0$ for all $k$, with condition numbers below $8$.

Table~\ref{tab:mimo_results} summarizes the results using $N = 9000$ data points. In the noise-free case, the Markov parameter error is at $O(10^{-14})$, and the Hankel singular values exhibit a clear gap at $Mn_p = 9$, confirming numerical separation of the $Mn_c = 6$ cancellable modes from the $Mn_p = 9$ plant modes. The exact extraction therefore extends to the MIMO case as predicted by Theorem~\ref{thm:plant_extraction}. At SNR $= 40$ dB, the CL fit is $93.8\%$ and the Markov error degrades to $O(10^{-2})$, reflecting the larger augmented dimension ($Mn_{cl} = 15$) and the increased number of parameters to be estimated from finite data, while the closed-loop model still captures the dominant input-output behavior.

\begin{table}[h]
\centering
\caption{Identification results for the MIMO plant (Example~3)}
\label{tab:mimo_results}
\begin{tabular}{ccc}
\hline
 & SNR $= \infty$ & SNR $= 40$ dB \\
\hline
CL fit [\%] & $100.0$ & $93.8$ \\
Max Markov error & $1.36 \times 10^{-14}$ & $5.93 \times 10^{-2}$ \\
\hline
\end{tabular}
\end{table}

The Markov-parameter error in Table~\ref{tab:mimo_results} serves as the coordinate-invariant primary criterion, while the CL fit complements it by confirming that the identified shared-$A$ closed-loop model represents the data-generating map from $r$ to $[y^{\top}\;u^{\top}]^{\top}$. These results confirm that the proposed framework extends to MIMO LPTV plants and solves Problem~\ref{prob:cl_id} in the coordinate-invariant sense, with the algebraic extraction of Theorem~\ref{thm:plant_extraction} remaining valid for the increased augmented dimension $Mn_{cl}=15$.

%================================================================================
\section{Conclusion}\label{sec:conclusion}
%================================================================================

This paper studied closed-loop identification of LPTV plants through cyclic reformulation, with emphasis on open-loop unstable plants for which open-loop experiments are not viable. The main result is an exact algebraic plant-extraction theorem: from a shared-$A$ realization of the cycled closed-loop maps from $\check r$ to $\check y$ and $\check u$, the cycled plant transfer matrix is recovered as $\check P(z)=\check T_{yr}(z)\check T_{ur}(z)^{-1}=\check{T}_{yr}^{(1)}(z)\,[\check{T}_{ur}^{(1)}(z)]^{-1}$ without increasing the state dimension, where the one-step-advanced transfer functions $\check{T}_{yr}^{(1)}$ and $\check{T}_{ur}^{(1)}$ provide the proper realizations used for the algebraic construction. The proof is purely algebraic, based on inverse-system realization and exact cancellation in the shared-$A$ structure, and remains valid when the recovered plant realization is unstable. Combined with subspace identification of the cycled closed-loop signals, this yields Algorithm~\ref{algo:cl_id}, which solves Problem~\ref{prob:cl_id} under Assumptions~\ref{ass:Dp_zero}--\ref{ass:obsv}.

The framework shifts the stability requirement for data generation from the open-loop plant to the internally stable closed-loop system, while retaining structural assumptions for the algebraic inversion: a square strictly proper plant, an invertible first Markov parameter of the controller path, and a minimal cycled closed-loop realization. The dependence of the extraction step on the conditioning of the inverse controller path was made explicit. This combination of the cyclic reformulation, the joint input-output closed-loop structure, and the shared-$A$ algebraic extraction is distinct from LTI closed-loop subspace identification, which does not exploit periodic structure; from LPV subspace identification, in which parameter variation is driven by an external scheduling signal rather than by time periodicity; from open-loop LPTV identification, which presupposes an open-loop stable plant; and from lifted closed-loop multirate frequency-response identification, which targets nonparametric frequency-response estimation rather than a finite-dimensional state-space LPTV plant model.

Numerical examples confirmed the exact extraction and the finite-data implementation for stable, open-loop unstable, and MIMO LPTV plants. The open-loop unstable example demonstrated the central feature of the proposed framework, namely that an unstable cycled plant realization is recovered from stable closed-loop data with the algebraic extraction remaining valid in the presence of unstable plant modes; the MIMO case confirmed that the framework extends to the increased augmented dimension and was validated by coordinate-invariant Markov-parameter comparisons. Future work will address relaxation of the structural assumptions, including extensions to non-square and non-strictly-proper plants, and a finite-sample noise-sensitivity analysis of the proposed extraction.

\section*{Declaration of Generative AI and AI-assisted technologies in the writing process}
During the preparation of this work the author used Claude (Anthropic) to refine English expressions. The author reviewed and edited the content and takes full responsibility for the manuscript.

%================================================================================
\section*{Funding}
This research did not receive any specific grant from funding agencies in the public, commercial, or not-for-profit sectors.

%================================================================================

\end{document}